\documentclass[12pt]{iopart}
\usepackage{amsthm} 
\usepackage{dcolumn}   
\usepackage{amsfonts,amssymb,graphics,graphicx,epsfig,color,times}
\expandafter\let\csname equation*\endcsname\relax
\expandafter\let\csname endequation*\endcsname\relax
\usepackage{amsmath}
\usepackage{dsfont}

\newtheorem{thm}{Theorem}
\newtheorem{cor}[thm]{Corollary}

\theoremstyle{remark}

\theoremstyle{definition}
\newtheorem{dfn}[thm]{Definition}
\newtheorem{example}{Example}

\newcommand{\tensormax}{\otimes_{\mathrm{max}}}
\newcommand{\PPR}[0]{P^{\text{PR}}}
\newcommand{\PNL}[0]{P^{\text{NL}}}

\newcommand{\Pc}[0]{P^{\text{L}}}

\newcommand{\bra}[1]{\left\langle#1\right|} 
\newcommand{\ket}[1]{\left|#1\right\rangle} 

\newcommand{\id}{\mathds{1}}

\newcommand{\wuerzburg}{\address{$^1$ Fakult{\"{a}}t f\"{u}r Physik und Astronomie, Universit\"{a}t W\"{u}rzburg, Am Hubland, 97074 W\"{u}rzburg, Germany}}
\newcommand{\potsdam}{\address{$^2$ Institute for Physics and Astronomy, Potsdam University, 14476 Potsdam, Germany}}
\newcommand{\bristolmath}{\address{$^3$ Department of Mathematics, University of Bristol, University Walk, Bristol BS8 1TW, United Kingdom}}
\newcommand{\rhul}{\address{$^4$ Department of Mathematics, Royal Holloway, University of London, Egham Hill, Egham TW20 0EX, United Kingdom}}
\newcommand{\bristol}{\address{$^5$ H.H. Wills Physics Laboratory, University of Bristol, Bristol BS8 1TL, United Kingdom}}

\begin{document}
\title{Limits on non-local correlations from the structure of the local state space}
\author{Peter Janotta$^1$, Christian Gogolin$^{1,2,3}$, Jonathan Barrett$^{4,5}$\\ and Nicolas Brunner$^5$}
\wuerzburg
\potsdam
\bristolmath
\rhul
\bristol
\eads{peter.janotta@physik.uni-wuerzburg.de}


\begin{abstract}
The outcomes of measurements on entangled quantum systems can be nonlocally correlated. However, while it is easy to write down toy theories allowing arbitrary nonlocal correlations, those allowed in quantum mechanics are limited. Quantum correlations cannot, for example, violate a principle known as macroscopic locality, which implies that they cannot violate Tsirelson's bound. This work shows that there is a connection between the strength of nonlocal correlations in a physical theory, and the structure of the state spaces of individual systems. This is illustrated by a family of models in which local state spaces are regular polygons, where a natural analogue of a maximally entangled state of two systems exists. We characterize the nonlocal correlations obtainable from such states. The family allows us to study the transition between classical, quantum, and super-quantum correlations, by varying only the local state space. We show that the strength of nonlocal correlations - in particular whether the maximally entangled state violates Tsirelson's bound or not - depends crucially on a simple geometric property of the local state space, known as strong self-duality. This result is seen to be a special case of a general theorem, which states that a broad class of entangled states in probabilistic theories - including, by extension, all bipartite classical and quantum states - cannot violate macroscopic locality. Finally, our results show that there exist models which are locally almost indistinguishable from quantum mechanics, but can nevertheless generate maximally nonlocal correlations.
\end{abstract}

\maketitle

\section{Introduction}

Nonlocality is a key feature of quantum mechanics. By performing measurements on separated systems in an entangled state, one can obtain correlations that are stronger than those of any local model, as witnessed by the violation of Bell inequalities \cite{bell64}. On the other hand, sets of nonlocal correlations are known that are stronger than those of quantum mechanics, but which do not allow for instantaneous signalling. This led Popescu and Rohrlich \cite{PR} to raise the question of why nonlocality seems to be limited in nature.

In recent years, new insights have been gained into this question by studying the information theoretic properties of super-quantum correlations. For instance these correlations lead to implausible reductions for all communication complexity problems, such that they can be solved with only constant communication \cite{vanDam,brassard,BS}. The principle of \emph{information causality} \cite{IC} is satisfied by quantum correlations, but can be violated if certain super-quantum correlations are available --- similarly the principle of \emph{macroscopic locality} \cite{ML}. Various multi-player games have been described, for which super-quantum correlations would provide an advantage over quantum correlations \cite{noah,GYNI}.

The above studies focused on the information theoretic power of correlations without any reference to the physical theories they emerge from.
Recent works revealed interesting connections between the structure of quantum mechanics and the nonlocal correlations that can be generated by quantum systems. Barnum et al. \cite{Beigi09}, for example, considered a theory that is locally equivalent to quantum mechanics but whose non-locality is only limited by the no-signalling principle. Despite this theory being less restrictive than quantum mechanics, the set of bipartite correlations that can be obtained is identical to that of quantum states. This implies that, despite the fact that quantum correlations are clearly a global property of joint systems, their limitation does not result from the lack of joint states, but rather from the structure of the local state spaces. Meanwhile, Ac\'{i}n et al. \cite{Acin10} have shown that this result does not extend to three or more parties.

In this paper, we show that the connection between local state spaces and the limitation of bipartite nonlocal correlations is actually a more general phenomenon. In particular, if local state spaces have a property known as \emph{strong self-duality}, then the correlations obtainable from maximally entangled states must be compatible with the principle of macroscopic locality. It follows that they must also respect Tsirelson's bound. A precise definition of strong self-duality is given later, but in the quantum case it corresponds roughly to the fact that the same rank one projector represents both a pure state and the outcome of a measurement which identifies that state.

\begin{figure}[bt]
 \centering
 \includegraphics[width=0.75 \linewidth]{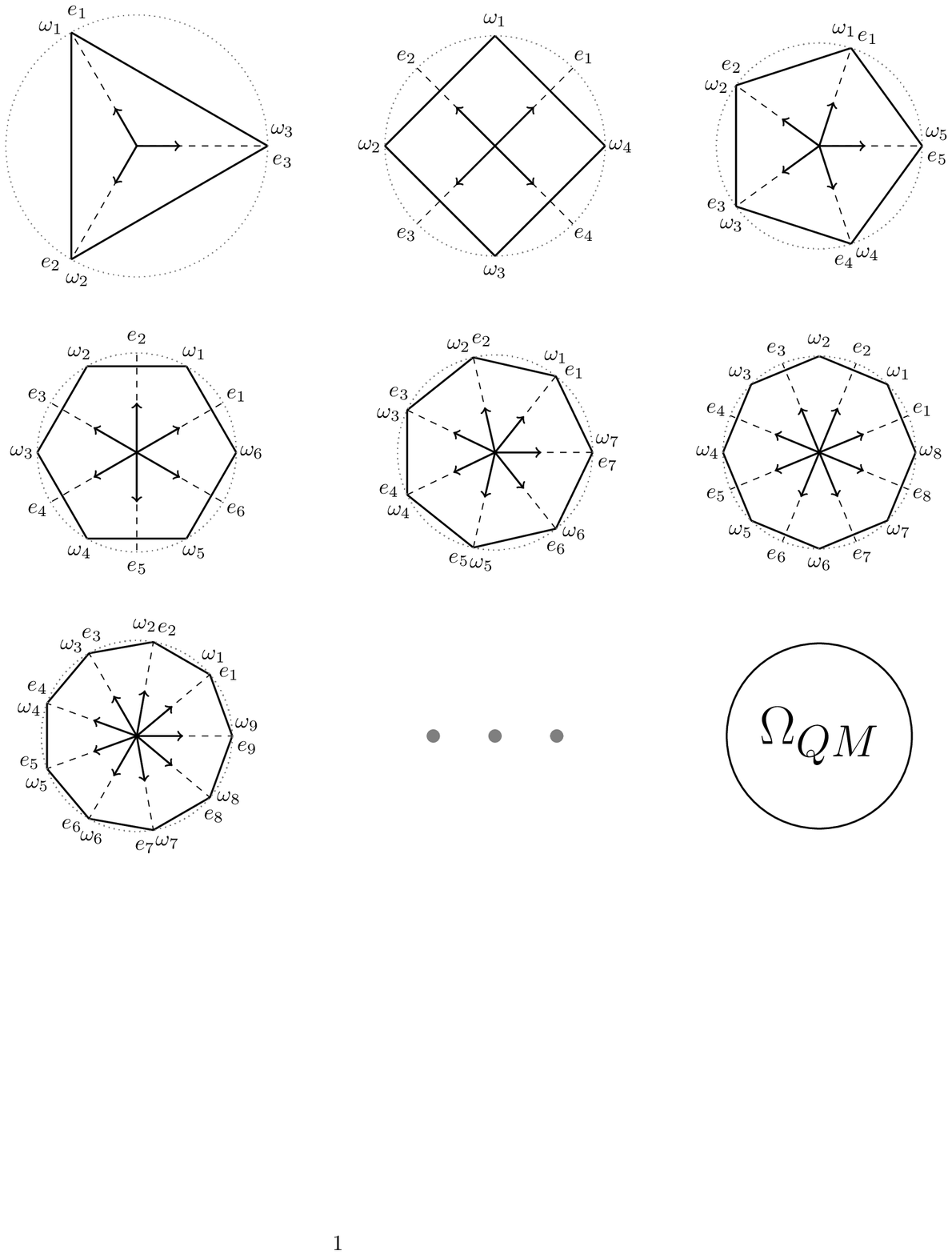}
 \caption{Illustration of the state spaces and ray extremal effects of the polygon models.}
 \label{model}
\end{figure}

By way of illustration, we introduce along the way a family of models, where each model is defined by the local state space for a single system, and the state space is taken to be a regular polygon with $n$ vertices (see \fref{model}). For two such systems, there is a natural analogue of a maximally entangled state. The family includes the classical case of two trits ($n=3$); systems generating the super-quantum correlations introduced by Popescu and Rohrlich ($n=4$); and systems producing quantum correlations ($n \rightarrow \infty$). Thus the family allows us to study the transition between these theories, and the bipartite correlations that can be produced by a maximally entangled state, by modifying only the local state space. For high $n$ the local state spaces are almost indistinguishable from a quantum system. Nevertheless it turns out that these models show dramatically different correlations --- and thereby have fundamentally different information theoretic capabilities --- depending on the parity of $n$. This is explained by the fact that those with odd $n$ are \emph{strongly self-dual}, while those with even $n$ only \emph{weakly self-dual}.

One way of viewing the polygon models is that moving from $n\rightarrow \infty$ to $n=3$, there is a progressive weakening of the superposition principle. A weakened superposition principle means that states can only be superposed in certain combinations. In a similar spirit, a different range of models was introduced in Ref.~\cite{UR}, with each model defined by a relaxation of the uncertainty relations of quantum mechanics. Here too, a transition from quantum correlations to Popescu-Rohrlich correlations was observed. 

This paper is organized as follows. Section \ref{operationalmodels} gives a brief, not too technical, introduction to a mathematical formalism in which a very broad range of  probabilistic theories can be expressed, including quantum theory and classical probability theory. Section \ref{sec:afamilyofmodels} introduces the polygon models, and by investigating the properties of bipartite correlations, sheds some light on the relation between these and the local state space structure. Section \ref{sec:selfdualityandtsirelsonsbound} returns to the general case and contains the proof of the main theorem, which establishes a rigorous limit on the nonlocal correlations obtainable from a broad class of bipartite states in general probabilistic theories. In particular, states obtainable by norm-preserving local transformations from what we call \emph{inner product states} cannot violate the principle of macroscopic locality. Section \ref{polygonsrevisited} provides a formal definition of strong and weak self-duality, and discusses consequences of the main theorem for the correlations in bipartite polygon systems. Section \ref{sec:Correlationsoutsideofq1} presents a strongly self-dual system in which a non-maximally entangled state gives rise to correlations that cannot be obtained from any inner product state. Finally, section \ref{discussion} discusses some open questions.

\section{Operational models}\label{operationalmodels}

\subsection{Systems and measurements}

This section describes briefly the framework of generalized probabilistic theories \cite{barrett}, using the notation and conventions of Ref.~\cite{Barnum}. The aim is to be able to describe theoretical models other than the classical and quantum theories, and for these two to be included as special cases.

We start by taking an operational point of view. A \emph{state} of a system is a mathematical object that defines the outcome probabilities for all the measurements that can possibly be performed on this system. The \emph{state space} $\Omega$ of a system is the set of states that it can be prepared in. 

By defining the operations of summation and multiplication by a real number on states, we can identify $p \omega_1 + (1-p) \omega_2$ as the probabilistic mixture obtained by preparing $\omega_1$ with probability $p$ and $\omega_2$ with probability $1-p$. The state space $\Omega$ is now a convex set, embedded in a real vector space $V$. For simplicity, assume that $\Omega$ is compact and finite dimensional. States that can be represented by convex combinations of other states are \emph{mixed} states. The extremal points of the state space $\Omega$ cannot be written in such a form, and are \emph{pure} states. For a quantum system, for example, $\Omega$ is the set of density operators on a Hilbert space, and the pure states are the rank one projectors. For a qubit, $\Omega$ is particularly easy to visualize, since it corresponds to the Bloch ball, with pure states on the surface of the ball. For a (finite-dimensional) classical system, $\Omega$ is the set of probability distributions over some finite sample space.

A measurement outcome is represented by an \emph{effect}, that is a map $e\colon \Omega \to [0,1]$, where $e(\omega)$ is the probability of obtaining the outcome $e$ when the measurement is performed on a system in the state $\omega$. Probabilities of measurement outcomes should respect probabilistic mixtures of states, meaning that $e[p\,\omega_1 + (1-p) \, \omega_2] = p\, e(\omega_1) + (1-p)\, e(\omega_2)$, i.e., the effects are affine maps. A special effect is the \emph{unit effect} $u$, which is uniquely defined such that $u(\omega)=1$ for all $\omega \in \Omega$. The unit effect represents a measurement with a single outcome that is certain to occur regardless of what the state is. An arbitrary measurement is a set of effects $\{e_i\}$ summing to the unit effect $\sum_i e_i = u$. This ensures that outcome probabilities of measurements sum to one.

The set of proper effects $E(\Omega) = \{e : 0 \leq e(\omega) \leq 1\ \forall \omega \in \Omega\}$ is the convex hull of the unit effect, the zero effect and a set of extremal effects. For a quantum system, if states are density operators on a Hilbert space, then effects can be identified with positive semidefinite operators on the Hilbert space, in such a way that outcome probabilities are given by the usual trace rule. Measurements correspond to positive operator-valued measures. For a classical system, effects can be identified with fuzzy indicator functions on the sample space, i.e., maps from the sample space into $[0,1]$.

\subsection{Unnormalized states}
It is frequently useful to work with unnormalized states. Given a state space $\Omega$ and effect space $E(\Omega)$, let $V$ be the linear span of $\Omega$. The linear span of $E(\Omega)$ is then the dual space $V^*$. Both $V$ and $V^*$ are real vector spaces. In the case of a quantum system, for example, $V$ is the linear span of the density operators, which is the set of all Hermitian operators on the corresponding Hilbert space. Similarly, $V^*$ is the linear span of the positive semidefinite operators, which is also the set of all Hermitian operators.

An unnormalized state is an element of $V$ of the form $r\,\omega$, with $r>0$ and $\omega\in\Omega$. The set of all unnormalized states is a cone denoted $V_+$. Similarly, an unnormalized effect is an element of $V^*$ of the form $r\,e$ for $r>0$ and $e \in E(\Omega)$. The set of unnormalized effects is the \emph{dual cone} to $V_+$, denoted $V_+^*$. The cone $V_+$ and the dual cone $V_+^*$ are related via
\begin{equation}
 V_+^* = \{ e \in V^*: e(\omega) \geq 0, \forall \omega \in V_+\}.
\end{equation}
In the case of a quantum system, both $V_+$ and $V_+^*$ can be identified with the set of positive semidefinite operators on the Hilbert space.
In general a cone $V_+$ can have a very different structure than its dual cone $V_+^*$, e.g., they may have a different number of extremal rays. 

\subsection{Bipartite states}\label{bipartitestates}
Given two systems $A$ and $B$, an operational model needs to specify the set $\Omega^{AB}$ of available joint states, in addition to the individual state spaces $\Omega^A$ and $\Omega^B$. In general, one can imagine many weird and wonderful ways in which two systems might combine to form a joint system. By imposing two quite natural conditions, however, one can narrow down these possibilities significantly.

The first condition is the \emph{no-signalling principle}, which says that it should not be possible to send messages instantaneously by performing measurements on the separate parts of a joint system. The second is that of \emph{local tomography}. Given a single system, call a measurement \emph{informationally complete} if its outcome probabilities are sufficient to determine uniquely the state of the system. The principle of local tomography states that if an informationally complete measurement is performed separately on each of the subsystems of a composite system, then the joint outcome probabilities are sufficient to determine uniquely the state of the joint system.  

These two conditions together are sufficient to ensure that the linear space $V^{AB}$ in which the joint state space $\Omega^{AB}$ and the cone of associated unnormalized states are embedded can be taken to be $V^A\otimes V^B$ (see for example Ref.~\cite{Barnum} and the references therein). If simultaneous measurements are performed on systems $A$ and $B$, then the joint probability for outcomes $e$ and $f$ is given by $(e\otimes f)(\omega^{AB})$.

It is convenient to define the unit effect of the joint state space as $u^{AB} = u^A \otimes u^B$ such that a joint state is normalized if
\begin{equation}\label{jointstatenorm}
(u^A \otimes u^B)(\omega^{AB})  = 1,
\end{equation}
where $u^A$ and $u^B$ are the unit effects for systems $A$ and $B$ respectively.
Naturally, probabilities are positive, so a joint state must satisfy
\begin{equation}\label{eq:tensormax}
(e^A\otimes e^B)(\omega^{AB}) \geq 0
\end{equation}
for all $e^A\in E(\Omega^A)$, $e^B\in E(\Omega^B)$.
\begin{dfn}
The \emph{maximal tensor product} of $\Omega^A$ and $\Omega^B$, denoted $\Omega^A \tensormax \Omega^B$, is the set of all $\omega^{AB} \in V^A\otimes V^B$ such that \eref{jointstatenorm} and \eref{eq:tensormax} are satisfied.
\end{dfn}

It is easy to check that the no-signalling principle is indeed satisfied for such an $\Omega^{AB}$.
Consider two measurements on $A$, corresponding to sets of effects $x = \{ e_1,\ldots, e_m \}$ and $x' = \{ e'_1,\ldots, e'_n \}$. The marginal probability for an outcome $f$ of a measurement on $B$ is
\begin{equation}
\sum_{i=1}^m (e_i\otimes f)(\omega^{AB}) = (u^A\otimes f)(\omega^{AB}) = \sum_{j=1}^n (e'_j\otimes f) (\omega^{AB}),
\end{equation}
i.e., it is independent of whether $x$ or $x'$ is performed on $A$.

Intuitively, the maximal tensor product is the set of all non-signalling joint states that can be written down for two systems, given the individual state spaces $\Omega^A$ and $\Omega^B$. A particular theory or model need not assume that every element of the maximal tensor product is an allowed state for the joint system. In general, a model will specify a joint state space $\Omega^{AB}$ which is a subset of $\Omega^A\tensormax\Omega^B$. 

Straightforwardly generalizing the notions well known from quantum theory, one calls a state a \emph{product state} if it can be written in the form $\omega^A \otimes \omega^B$ for some states $\omega^A \in \Omega^A$ and $\omega^B \in \Omega^B$. States that can be written as probabilistic mixtures of product states are \emph{separable}, while states that are not separable are \emph{entangled}. 

This work mostly considers correlations obtained from product measurements on bipartite states. The general formalism, however, does not assume that all measurements on composite systems are product measurements. As in the case of single systems, outcomes of measurements on a composite system correspond to effects, where these are maps $\Omega^{AB}\rightarrow [0,1]$. The set of all such effects is written $E(\Omega^{AB})$, and may include entangled, as well as product, effects. However, $E(\Omega^A\tensormax\Omega^B)$ only contains separable effects.

Quantum theory provides a useful example of many of the concepts above. In this case, $\Omega^{AB}$ is the set of density operators on the Hilbert space $H^{AB} = H^A \otimes H^B$. Recall that $V^A$ and $V^B$ are real vector spaces of Hermitian operators on $H^A$ and $H^B$ respectively. The set of Hermitian operators on $H^{AB}$ can be identified with $V^A\otimes V^B$, so the joint quantum states are indeed elements of $V^A\otimes V^B$. The density operators on $H^{AB}$ are a proper subset of $\Omega^A\tensormax\Omega^B$. Elements of $\Omega^A\tensormax\Omega^B$ which are not density operators are (normalized) \emph{entanglement witnesses}. An entanglement witness $w$ is locally positive, meaning that for all product measurements, $(e^A\otimes e^B)(w) \geq 0$. But $w$ is not a density operator, since there are entangled measurement outcomes $e$ with $e(w) < 0$.

\section{A family of models}\label{sec:afamilyofmodels}

\subsection{Polygon systems}\label{polygonsystems}
This section defines a family of models such that the state spaces $\Omega$ of single systems are regular polygons with $n$ vertices. It is convenient to represent both states and effects by vectors in $\mathbb{R}^3$ such that $e(\omega)$ is the usual Euclidean inner product. For fixed $n$, let $\Omega$ be the convex hull of $n$ pure states $\{\omega_i\}$, $i=1,...,n$, with
\begin{equation}\label{eq:localpolygons}
 \omega_i = \begin{pmatrix}
             r_n \cos(\frac{2 \pi i}{n})\\
             r_n \sin(\frac{2 \pi i}{n})\\
             1
            \end{pmatrix} \in \mathbb{R}^3 ,
\end{equation}
where $r_n= \sqrt{\sec(\pi/n)}$.

The unit effect is
\begin{equation}
   u = \begin{pmatrix}
        0\\
        0\\
        1
       \end{pmatrix}.
\end{equation}
In the case of even $n$, the set $E(\Omega)$ of all possible measurement outcomes is the convex hull of the zero effect, the unit effect, and $e_1,\ldots, e_n$, with
\begin{equation}
  \label{eff_even}
  e_i = \frac{1}{2} \, \begin{pmatrix}
    r_n \cos(\frac{(2 i-1) \pi}{n})\\
    r_n \sin(\frac{(2 i-1) \pi}{n})\\
    1
  \end{pmatrix}  .
\end{equation}
Let $\bar{e_i} = u - e_i$, hence a possible dichotomic measurement is $\{e_i, \bar{e_i} \}$. When this measurement is performed on a system in the state $\omega_j$, the probabilities for the two outcomes are given by $e_i\cdot \omega_j$ and $\bar{e_i}\cdot \omega_j$, and satisfy $e_i\cdot \omega_j + \bar{e_i}\cdot \omega_j = 1$. Observe that for even $n$, $\bar{e_i} = e_{(i+n/2) \mathrm{mod} \ n}$.

The case of odd $n$ is slightly different. In this case, define
\begin{equation}
  \label{eff_odd}
  e_i = \frac{1}{1 + {r_n}^2} \, \begin{pmatrix}
    r_n \cos(\frac{2 \pi i}{n})\\
    r_n \sin(\frac{2 \pi i}{n})\\
    1
  \end{pmatrix} 
\end{equation}
and again let $\bar{e_i} = u - e_i$, so that a possible dichotomic measurement is $\{e_i, \bar{e_i} \}$. This time, however, $\bar{e_i}$ does not equal $e_j$ for any $j$. The set $E(\Omega)$ of all possible measurement outcomes is the convex hull of the zero effect, the unit effect, $e_1\ldots,e_n$, and $\bar{e_1},\ldots,\bar{e_n}$.
As can be seen in \fref{fig:poly3d} in such theories there are effects that are extremal in $E(\Omega)$ (namely the $\bar{e_i}$) but not ray extremal, i.e., they do not lie on an extremal ray of the cone $V_+^*$.
This also happens in quantum mechanics, but only if the dimension of the Hilbert space is larger than two.
For example the effect $\id - | \psi \rangle\langle \psi |$ for any rank one projector $| \psi \rangle\langle \psi |$ is then extremal in the set of proper effects, but not ray extremal.

A two-dimensional illustration of the state and effect spaces is given in \fref{model} and a three-dimensional illustration in \fref{fig:poly3d}. 
\begin{figure}[tb]
 \centering
 \includegraphics[angle=-90,width=0.75\linewidth]{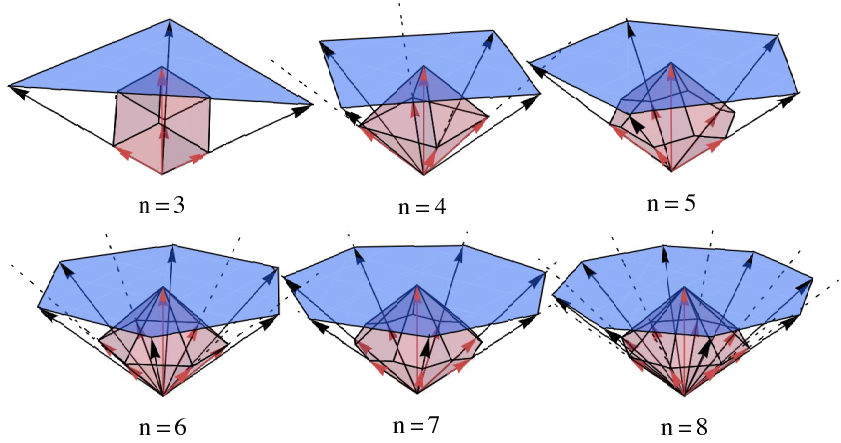}
 \caption{State spaces $\Omega$ (blue polygons) and sets of proper effects $E(\Omega)$ (red polytopes) of the polygon toy theories with $n$ vertices. The case $n=3$ corresponds to a classical system, the $n=4$ system is capable of generating all no-signalling correlations. In the limit $n\to\infty$ the state space becomes a disc, which can be thought of as the equatorial plane of the Bloch ball.}
 \label{fig:poly3d}
\end{figure}

The $n=3$ case corresponds to a classical system with three pure states. Think of it as a trit. The three pure states are $\omega_1$, $\omega_2$ and $\omega_3$, and correspond to the three different possible values of the trit. The state space $\Omega$ is a triangle. A generic point in $\Omega$ is a mixture of the three pure states and corresponds to a probability distribution over the three trit values. Notice that in this case, $e_1+e_2+e_3 = u$, hence a possible measurement is a three-outcome measurement with outcomes $e_1,e_2$ and $e_3$. This is the obvious measurement that simply reads off the value of the trit. Below we shall consider bipartite states of polygon systems. Given two trits, the only possible joint states are separable, and it is not possible to produce nonlocal correlations.
The case $n=4$ corresponds to a single system in a toy theory known as `box world', which has been discussed elsewhere in the literature (see for instance Ref.~\cite{barrett}). The state space is a square. As shown below, a notable feature of box world is that given two of these systems, it is possible to construct joint states that are more nonlocal than quantum states. In fact, an entangled state of two of the $n=4$ systems can produce maximally nonlocal correlations known as \emph{PR box} correlations \cite{PR}, which have been much explored in the literature \cite{vanDam, brassard, IC, noah}.

As $n \rightarrow\infty$, the state space tends to a disc of radius one. This makes it similar to a quantum mechanical qubit, whose state space is the Bloch ball. The disc can be thought of as the equatorial plane of the Bloch ball. We will refer to this case, somewhat loosely, as the quantum case.

\subsection{Bipartite states of polygon systems}
\label{sec:defpolybox}
We shall not attempt a complete characterization of the set of all possible non-signalling states $\Omega^A \tensormax \Omega^B$ for each value of $n$. Instead, this section describes a particular joint state of two polygon systems, which is the natural analogue of a maximally entangled state of two qubits. The next section examines the nonlocal correlations that can be obtained from performing measurements on these maximally entangled polygon systems.

Recall that a joint state is an element of $V^A\otimes V^B$, hence in the case of two polygon systems, a joint state is an element of $\mathbb{R}^3 \otimes \mathbb{R}^3 = \mathbb{R}^9$. It is convenient to represent the joint state as a $3\times3$ matrix such that $(e_i \otimes e_j)(\omega^{AB})$ can be calculated by simply left and right multiplying this matrix with the representations of the effects $e_i$ and $e_j$ in $\mathbb{R}^3$. Define
\begin{eqnarray}\label{definitionphi}
\mathrm{odd \ n:}\quad\phi^{AB} &=& \left(\begin{array}{ccc} 1 & 0 & 0 \\ 0 & 1 & 0 \\ 0 & 0 & 1 \end{array}\right), \nonumber\\
\mathrm{even \ n:}\quad\phi^{AB} &=& \left(\begin{array}{ccc} \cos(\pi/n) & \sin(\pi/n) & 0 \\ -\sin(\pi/n) & \cos(\pi/n) & 0 \\ 0 & 0 & 1 \end{array}\right). 
\end{eqnarray}

The state $\phi^{AB}$ is the natural analogue of a quantum mechanical maximally entangled state for the following reasons. First, it can be verified (see, e.g., Ref.~\cite{steering}) that except for $n=3$, $\phi^{AB}$ is an entangled pure state, where pure means that it is extremal in the maximal tensor product, hence cannot be written as a mixture of other non-signalling states. The $n=3$ case corresponds to two classical trits, with $\phi^{AB}$ the maximally correlated state, i.e., if the trit values are $1$, $2$, $3$, then $\phi^{AB}$ corresponds to $P(11)=P(22)=P(33)=1/3$. Second, $\phi^{AB}$ is constructed so that if a measurement is performed on the $A$ system, and outcome $e_i$ obtained, then the updated (or collapsed) state for the $B$ system is $\omega_i$. The marginal probability for Alice to obtain outcome $e_i$ is the same for all $i$. Compare this with the case of two spin-1/2 particles in the state $1/\sqrt{2} (|00\rangle + |11\rangle)$, where $|0\rangle$ and $|1\rangle$ are the eigenstates of spin-$z$. If a spin measurement in direction $\vec{m}$ in the $xz$-plane is performed on system $A$, then the probability of obtaining the up outcome is $1/2$, and if the up outcome is obtained, then the collapsed state of the $B$ system is spin up in direction $\vec{m}$. These quantum predictions are recovered by $\phi^{AB}$ in the limit $n \rightarrow \infty$.

The following sections investigate the nonlocal correlations that can be produced by performing measurements on two systems in the state $\phi^{AB}$. For this it is useful to have an expression for the joint probability of obtaining outcome $e^A_i$ on system $A$ and $e^B_j$ on system $B$. This is easy to calculate from \eref{definitionphi}. For even $n$,
\begin{equation}\label{eq:nevencorrelations}
(e^A_i \otimes e^B_j)(\phi^{AB}) = \frac{1}{4}\left( 1+r_n^2 \cos(\alpha_i-\beta_j)\right),
\end{equation}
where $\alpha_i = \frac{2\pi i}{n}$ and $\beta_j = \frac{(2j-1)\pi}{n}$, and as before, $r_n = \sqrt{\sec(\pi/n)}$. For odd $n$
\begin{equation}\label{eq:noddcorrelations}
(e^A_i \otimes e^B_j)(\phi^{AB}) = \frac{1}{(1+r_n^2)^2}\left( 1+r_n^2 \cos(\alpha_i-\beta_j)\right),
\end{equation}
where $\alpha_i = \frac{2\pi i}{n}$ and $\beta_j = \frac{2\pi j}{n}$. Notice the cosine dependence, which is reminiscent of quantum mechanical correlations.

\subsection{The Clauser-Horne-Shimony-Holt inequality}\label{sectionchsh}
One commonly used measure of the degree of nonlocality that a bipartite system exhibits is the maximal violation of the Clauser-Horne-Shimony-Holt (CHSH) inequality \cite{chsh}. The CHSH inequality involves two parties, conventionally called Alice and Bob. Each chooses between two dichotomic measurements. Let Alice's choice of measurement be $x$, and Bob's $y$, with  $x,y \in \{0,1\}$. Denote the measurement outcomes $a,b \in \{0,1\}$. A set of correlations is characterized by the joint probability distribution $P(a,b|x,y)$. The strength of the correlations is quantified by the CHSH parameter
\begin{equation}\label{chshparameter}
S = | E_{0,0}+ E_{0,1} + E_{1,0} - E_{1,1}| ,
\end{equation}
where $E_{x,y} = P(0,0|x,y)+P(1,1|x,y)-P(0,1|x,y)-P(1,0|x,y)$. As CHSH showed, local correlations must satisfy $S\leq 2$. In quantum mechanics, correlations can violate this inequality, but must respect Tsirelson's bound $S\leq 2 \sqrt{2}$ \cite{cirelson80}. 

By inspection, the algebraic maximum of $S$ is $4$, and it is easy to see that it is attained by the following correlations:
\begin{equation}\label{PR1}
P(a,b|x,y) =
\begin{cases}
\frac{1}{2} & \text{if $ a \oplus b = xy$} \\
0 & \text{otherwise}.
\end{cases}
\end{equation}
Here, $\oplus$ denotes addition modulo $2$.
These correlations were described by Popescu and Rohrlich, who pointed out that they are maximally nonlocal, yet still respect the no-signalling principle \cite{PR}. Since they cannot occur in quantum mechanics, they are imagined to be produced by a fictitious device, which is often referred to as a \emph{PR box}. As discussed in the introduction, PR boxes have been explored in the literature and are known to be particularly powerful for certain kinds of information theoretic problem, especially communication complexity problems \cite{vanDam,brassard,BS,IC,ML,noah,GYNI}.

It is interesting to see how the maximal CHSH value obtainable from polygon systems in the state $\phi^{AB}$ varies as the number of vertices $n$ of the polygon increases. The $n=4$ case is particularly simple. The optimal choice of measurements to violate the CHSH inequality is 
\begin{align}
  x = 0:& \{ e_1^A,e_3^A \}, & x = 1:& \{ e_2^A,e_4^A \}, & y = 0:& \{ e_2^B,e_4^B \}, & y = 1:& \{ e_1^B,e_3^B \}, 
\end{align}
and it can be verified from \eref{eq:nevencorrelations} that the correlations obtained give $S=4$. In other words, the maximally entangled state of two $n=4$ systems can act as a PR box. It follows that this state has the same information theoretic power that PR boxes are known to have. 

For general $n$, assume that Alice's measurement choices are of the form $\{ e^A_i, \bar e^A_i\}$ and Bob's of the form $\{ e^B_j,\bar e^B_j\}$. A lengthy but straightforward calculation gives the following analytic expressions. 
For even $n$,
\begin{equation}\label{CHSHeven}
S = r_n^2 \sum_{x,y=0,1} (-1)^{xy} \cos\left(\alpha_x-\beta_y\right),
\end{equation}
where as before, $\alpha_x = \frac{2\pi i_x}{n}$ and $\beta_y = \frac{(2j_y-1)\pi}{n}$.
For odd $n$,
\begin{equation}\label{CHSHodd}
S = \frac{2}{\left(1+r_n^2 \right)^2} \, \bigg| (r_n^2-1)^2 + 2\,r_n^2  \sum_{x,y=0,1}(-1)^{xy} \cos(\alpha_x-\beta_y) \bigg|,
\end{equation}
where $\alpha_x = \frac{2\pi i_x}{n}$ and $\beta_y = \frac{2\pi j_y}{n}$.
Maximizing these expressions over all possible choices for the angles $\alpha_i$ and $\beta_j$ gives the maximal violation achievable by local measurements on the maximally entangled state $\phi^{AB}$.
A detailed analysis of these expressions can be found in \ref{app:optimalchshsvalue}. \Fref{fig:chshpoly} shows the maximal CHSH value for the maximally entangled state of polygon systems as a function of $n$. 
\begin{figure}[ht]
 \centering
 \includegraphics[width=0.65 \linewidth]{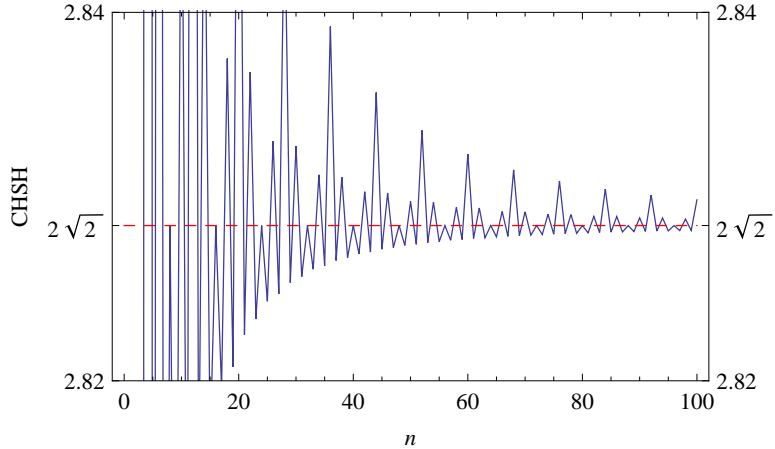}
 \caption{Maximal CHSH value from the maximally entangled state of two polygon systems as a function of the number of vertices $n$. Tsirelson's bound ($S \leq 2\sqrt{2}$) appears as a natural separation between the case of even $n$ and odd $n$.}
 \label{fig:chshpoly}
\end{figure}

The most important feature of \fref{fig:chshpoly} is that the correlations of even $n$ systems can always reach or exceed Tsirelson's bound, while the correlations of odd $n$ systems are always below Tsirelson's bound. Thus Tsirelson's bound appears as a natural separation between the correlations of these two different kinds of polygon state spaces. Sections \ref{sec:selfdualityandtsirelsonsbound} and \ref{polygonsrevisited} show why this is. Section \ref{sec:selfdualityandtsirelsonsbound} shows that for odd $n$, the maximally entangled state $\phi^{AB}$ belongs to a broad class of states we call \emph{inner product states}, and that all correlations obtainable from measurements on inner product states satisfy Tsirelson's bound. Section \ref{polygonsrevisited} goes further, and relates this to a fundamental geometric difference between polygons with even $n$ and odd $n$. In \fref{model}, the difference is seen in the fact that for odd $n$, the effect cone $V_+^*$ coincides with the state cone $V_+$, whereas for even $n$, the effect cone is isomorphic to the state cone but rotated through some angle. 

We have only considered correlations obtainable from the maximally entangled state $\phi^{AB}$. In principle there could be joint states other than the maximally entangled state which show stronger violations for some Bell inequalities. While this seems unlikely for the CHSH inequality, other Bell inequalities are known to be maximized by non-maximally entangled states in quantum mechanics \cite{nonmax}.

\subsection{The Braunstein-Caves inequalities}\label{sectionbc}
The Braunstein-Caves (or \emph{chained}) Bell inequalities \cite{chained} are similar to the CHSH inequality, but involve $N$ measurement settings on each system, rather than two. Let Alice's choice of measurement be $x$, and Bob's $y$, with $x,y \in \{1,\ldots, N \}$. Let the outcomes be $a,b \in \{0,1\}$. Local correlations satisfy
\begin{equation}
  S_N = \bigg|\sum_{j=1}^{N-1} (E_{j,j}+E_{j,j+1})  + E_{N,N} -E_{N,1}  \bigg| \leq 2N-2,
\end{equation}
where as before $E_{x,y} = P(0,0|x,y)+P(1,1|x,y)-P(0,1|x,y)-P(1,0|x,y)$. In the case $N=2$, this is equivalent to the CHSH inequality, up to relabelling of measurement settings.

The algebraic maximum of $S_N$ is $2N$. This maximum can be attained by performing measurements on the maximally entangled state of even $n$ polygon systems with $n=2N$. This state is thus tailor made for violating the Braunstein-Caves Bell inequalities. To see this, let Alice's and Bob's measurement choices be given by
\begin{align}
x &= i: \{ e^A_i, \bar e^A_i \}, \quad i=1,\ldots, N, \\
y &= j: \{ e^B_j, \bar e^B_j \}, \quad j=1,\ldots, N,
\end{align}
and note that (i) $E_{j,j}=1$ for $j=1,...,N$, (ii) $E_{j,j+1}=1$ for $j=1,...,N-1$ and (iii) $E_{N,1} = -1$. In the case $n\rightarrow \infty$, maximal violation of the Braunstein-Caves inequality is achieved in the limit of infinitely many settings. This is also true for a quantum mechanical maximally entangled state, as shown in Ref.~\cite{barrettkentpironio}.

In general, given a set of correlations $P(a,b|x,y)$, they can be written as a mixture 
\begin{equation}
P(a,b|x,y) = q \PNL(a,b|x,y) + (1-q) \Pc(a,b|x,y),  
\end{equation}
where $0\leq q \leq 1$, $P^{NL}(a,b|x,y)$ is a set of nonlocal correlations and $P^L(a,b|x,y)$ a set of local correlations. Suppose, however, that the correlations $P(a,b|x,y)$ return the maximum value $S_N$ for an appropriate Braunstein-Caves inequality. Then $q(S_N) +(1-q)(S_N-2)\geq S_N$, hence $q=1$. Therefore, the fact that the maximally entangled state of even $n$ polygon systems returns the maximum value for the appropriate Braunstein-Caves inequality indicates that there is no local part in the correlations with $N=n/2$ measurement settings. This was pointed out in the case of quantum systems in Ref.~\cite{barrettkentpironio,epr2}. As a further curiosity, if we did have access to these systems, they could be used for secure key distribution, using the protocol of Ref.~\cite{BHK}.

\subsection{Distillation}
\label{sec:distill}
So far, we have only considered correlations that can be produced by measuring a single copy of a bipartite polygon system. There remains the possibility that stronger correlations could be produced by performing local measurements on multiple bipartite pairs, and locally processing the data (there is a further possibility, involving entangled measurements across multiple copies on each side, which we do not discuss).

Consider the bipartite state $\phi^{AB}$ of two even $n$ polygon systems, and suppose that Alice and Bob are choosing from the measurements
\begin{equation}
x = 0: \{ e_1^A,\bar e_1^A \}, \quad x = 1: \{ e_2^A,\bar e_2^A \}, \quad y = 0: \{ e_1^B,\bar e_1^B \}, \quad y = 1: \{ e_2^B,\bar e_2^B \}, 
\end{equation}
with outcomes $a,b \in \{0,1\}$ as usual.
Recall that $E_{j,j}=1$ for $j=0,1$ and $E_{0,1}=1$. \Eref{eq:nevencorrelations} also gives $E_{1,0}=2\cos(\frac{2\pi}{n})-1$. The correlations produced can be written as a probabilistic combination of maximally nonlocal correlations (equivalent up to relabelling to the PR box correlations of \eref{PR1}), and another term which describes local correlations:
\begin{equation}\label{NLC}
  P_\epsilon(a,b|x,y) = \epsilon \PPR(a,b|x,y) + (1-\epsilon) \Pc(a,b|x,y).
\end{equation}
Here, $ 0 \leq \epsilon=1-\cos(\frac{2\pi}{n}) \leq 1$, $\PPR$ is given by
\begin{equation}\label{PR}
\PPR (a,b|x,y) =
\begin{cases}
\frac{1}{2} & \text{if $ a \oplus b = x(y \oplus 1)$} \\
0 & \text{otherwise}
\end{cases}
\end{equation}
and $\Pc$ is a set of local correlations given by
\begin{equation}\label{Pc}
\Pc (a,b|x,y) =
\begin{cases}
\frac{1}{2} & \text{if $ a \oplus b = 0$,} \\
0 & \text{otherwise.}
\end{cases}
\end{equation}
In Ref.~\cite{BS}, it is shown that all correlations of the form \eqref{NLC} with $0<\epsilon <1$ can be distilled into stronger correlations using a protocol that involves two copies of a bipartite system. Importantly, this protocol consists only of local processing and does not involve any communication. In the asymptotic limit of infinitely many copies of a bipartite system, the correlations \eqref{NLC} can be distilled to PR box correlations by iterating the protocol. Thus for any finite even $n$, the polygon systems produce correlations that can be distilled arbitrarily close to PR box correlations (since $\epsilon = 1-\cos(\frac{2\pi}{n})>0$). It is only in the limit $n \rightarrow \infty$ (the quantum case), that we get $\epsilon=0$ and thus lose the ability to distill PR box correlations.

The consequence of the above is that polygon systems with even and finite $n$ inherit the powerful communication properties of PR boxes as long as there are multiple copies of the maximally entangled state available. For instance, they collapse communication complexity \cite{vanDam}, allow for better than classical non-local computation \cite{noah}, violate information causality \cite{IC} and macroscopic locality \cite{ML}. Moreover, since the PR box can be considered as a unit of bipartite nonlocality \cite{unit,forster2}, it follows that any bipartite no-signalling probability distribution can be generated from multiple copies of polygon systems with even $n$. This is particularly surprising as in practice, an individual polygon system with even and very large $n$ would be very difficult to distinguish from one with odd $n$, and also from the quantum case, i.e. the disc that one gets in the limit $n\to\infty$. These toy theories thus show that practically indistinguishable theories can have fundamentally different limits to the non-local correlations they allow.

For polygon systems with odd and finite $n$, the situation is dramatically different, as seen in the next section.

\section{Bounds on correlations}\label{sec:selfdualityandtsirelsonsbound}
For even $n$ polygon systems, the maximally entangled state can produce arbitrarily strong nonlocal correlations, whereas for odd $n$ polygon systems, the nonlocality is highly constrained. The maximally entangled state of odd $n$ polygon systems cannot, for example, violate Tsirelson's inequality. This section shows that this is a consequence of a much more general result. 

We first introduce a class of bipartite states in general theories, which we call \emph{inner product states}. The main theorem establishes a strong constraint on the nonlocal correlations that can be produced from measurements on inner product states. One consequence is that inner product states cannot violate Tsirelson's inequality. The maximally entangled states of odd $n$ polygon systems are inner product states, hence the theorem explains what was only established by direct calculation above --- that these states do not violate Tsirelson's inequality. On the other hand, the maximally entangled states of even $n$ polygon systems are not inner product states, which is consistent with them producing arbitrary non-signalling correlations. We also show that all classical and quantum states are, in terms of non-local correlations, no stronger than an inner product state.

\subsection{Inner product states}\label{sec:innerproductstates}

Recall that a state cone $V_+$ is the set of unnormalized states of a system, and that these span a vector space $V$. An effect cone $V_+^*$ is the set of unnormalized measurement outcomes, and these span the vector space $V^*$. Given two systems $A$ and $B$, if the state cones $V^A_+$ and $V^B_+$ span vector spaces $V^A$ and $V^B$ respectively, then a joint state is an element of $V^A\otimes V^B$.

Call two distinct systems \emph{similar} if their state spaces are isomorphic. Examples of similar systems are two quantum mechanical qubits, or two classical trits, or two $n$-vertex polygon systems. For the rest of this section, assume a bipartite system composed of two similar subsystems $A$ and $B$. In this case, the respective state spaces and effect spaces can be identified, so that $V^A=V^B=V$, $(V^A)^*=(V^B)^*=V^*$, $u^A=u^B=u$, and so on. 
\begin{dfn}
A joint state $\omega^{AB}$ is \emph{symmetric} if $(e\otimes f)(\omega^{AB})=(f\otimes e)(\omega^{AB})$ for all measurement outcomes $e,f\in V_+^*$.
\end{dfn}
\begin{dfn}
A joint state $\omega^{AB}$ is an \emph{inner product state} if $\omega^{AB}$ is symmetric, and positive semidefinite, i.e., $(e\otimes e)(\omega^{AB}) \geq 0 \ \forall e \in V^*$.
\end{dfn}
Note that by definition of a joint state, it is always true that $(e\otimes e)(\omega^{AB}) \geq 0$ when $e\in V_+^*$, i.e., when $e$ is a valid effect. This is simply a statement of the fact that measurement outcome probabilities have to be greater than or equal to zero. The definition requires something stronger, which is that $(e\otimes e)(\omega^{AB}) \geq 0$ for any $e$ in the whole of the vector space $V^*$. 

\begin{example}
Any symmetric product state $\omega^{AB} = \omega \otimes \omega$ is an inner product state.
\end{example}
\begin{example}
Consider two classical systems, each of which is a \emph{nit}, taking values $\{1,\ldots, n\}$. A joint state is simply a joint probability distribution over nit values. 
Write the joint state as a matrix $P$, where $P_{ij}$ is the joint probability that $A=i$ and $B=j$. This is an inner product state iff the matrix $P$ is symmetric and positive semi-definite. In particular this includes any perfectly correlated state of the form
\begin{eqnarray*}
P_{ij} &=& 0 \quad\mathrm{if}\quad i\ne j\\
P_{ii} &=& q_i, \quad q_i \geq 0, \quad \sum_i q_i = 1.
\end{eqnarray*}
\end{example}
\begin{example}
Consider two polygon systems, each corresponding to a state space with $n$ vertices. Section \ref{sec:defpolybox} defined an analogue of a maximally entangled state $\phi^{AB}$. In the matrix representation of \eqref{definitionphi}, $\phi^{AB}$ is an inner product state if and only if the matrix is symmetric and positive semi-definite. Hence $\phi^{AB}$ is an inner product state for odd $n$, whereas for even $n$, $\phi^{AB}$ is not an inner product state.
\end{example}
\begin{example}\label{quantumexample}
The quantum case is slightly subtle. Given two qubits, the maximally entangled state 
\begin{equation}
\Phi^+ = \ket{\Phi^+}\bra{\Phi^+},\qquad \ket{\Phi^+} = \frac{1}{\sqrt{2}} \left( \ket{00} + \ket{11} \right)
\end{equation}
is symmetric but is not an inner product state, since if $\sigma_y$ is a Pauli spin matrix, then $(\sigma_y\otimes \sigma_y)(\Phi^+) = -1$. Consider the operator defined by $\tilde{\Phi} = (\id\otimes T)(\Phi^+)$, where $T$ is the linear map that takes an operator in $V^B$ to its transpose with respect to the computational basis. The new operator $\tilde{\Phi}$ is not a valid quantum state. It is locally positive but not globally positive, hence is not a density operator. But it is in the maximal tensor product of two qubits, and it is an inner product state. In fact, $\tilde{\Phi}$ predicts perfect correlation whenever Alice and Bob perform measurements in the same direction.
However, the two states are equivalent in terms of the non-local correlations they can produce (as was first shown in Ref.~\cite{Beigi09}).
\end{example}

Theorem \ref{corr:identitystatecorrelationsareinq1} below establishes a constraint on the nonlocal correlations that can be obtained from measurements on an inner product state. It may seem as if the definition of an inner product state is quite restrictive, given that an inner product state must be symmetric, for example, and given that the maximally entangled state $\Phi_+$ of two qubits is not included. This would diminish the interest of the theorem. However, suppose that a bipartite state $\omega^{AB}$ can be obtained from an inner product state via a transformation of one of its subsystems. Then any correlations obtained from $\omega^{AB}$ could also be obtained from an inner product state. Hence any restriction on the correlations from inner product states also applies to $\omega^{AB}$. Formally,  
\begin{thm}\label{equivstatesweaker}
Consider a joint state $\omega^{AB}$, which can be written in the form $\omega^{AB} = (\id\otimes \tau)(\sigma^{AB})$, for some $\tau: V_+\rightarrow V_+$ that takes normalized states to normalized states. Any correlations obtained from measurements on $\omega^{AB}$ can also be obtained from measurements on $\sigma^{AB}$.
\end{thm}
\begin{proof}
Define the adjoint map $\tau^{\dagger}: V_+^*\rightarrow V_+^*$ such that for any effect $e\in V_+^*$ and any state $\omega\in V_+$, 
\begin{equation}
(\tau^\dagger(e))(\omega) = e(\tau(\omega)). 
\end{equation}
Since $\tau$ takes normalized states to normalized states, $\tau^{\dagger}(u) = u$. Given a measurement $y$ on system $B$, with outcomes $\{f_1,\ldots, f_r\}$, let $y'$ be the measurement with outcomes $\{\tau^{\dagger}(f_1), \ldots, \tau^{\dagger}(f_r)\}$. Note that from $f_1+\cdots+ f_r = u$, and $\tau^{\dagger}(u) = u$, it follows that $\tau^{\dagger}(f_1) + \cdots + \tau^{\dagger}(f_r) = u$, as must be the case for $y'$ to be a valid measurement. Then measurements $x$ and $y$ on $\omega^{AB}$ have the same joint outcome probabilities as measurements $x$ and $y'$ on $\sigma^{AB}$. Hence, if a particular set of correlations can be obtained by performing measurements on $\omega^{AB}$, those same correlations can be obtained by performing different measurements on $\sigma^{AB}$.
\end{proof}

Further,
\begin{thm}\label{quantumstatesareinnerprod}
Given two $d$-dimensional quantum systems, any pure state $\rho^{AB}=\ket{\psi}\bra{\psi}$ can be written in the form $\rho^{AB} = (\id \otimes \tau)(\tilde{\rho}^{AB})$, where $\tau: V_+\rightarrow V_+$ takes normalized states to normalized states, and $\tilde{\rho}^{AB}$ is an inner product state.
\end{thm}
\begin{proof}
Using the Schmidt decomposition, every pure quantum state $\ket{\psi}$ can be written in the form:
\begin{equation}\label{eq:schmidtdecomposition}
\ket{\psi} = \sum_{i=1}^r \lambda_i \ket{a_i} \otimes \ket{b_i},
\end{equation}
where $r$ is the Schmidt rank, $\{\ket{a_i}\}$ and $\{\ket{b_i}\}$ are orthonormal bases and the $\lambda_i$ are real and positive. A unitary transformation $U$, on system $B$, which maps $\{\ket{b_i}\}$ to $\{\ket{a_i}\}$ gives
\[
\ket{\psi'} = \sum_{i=1}^r \lambda_i \ket{a_i} \otimes \ket{a_i}.
\]
Now let 
\[
\tilde{\rho}^{AB} = (\id \otimes T)(\ket{\psi'}\bra{\psi'}),
\]
where $T$ is the transpose map, acting on the $B$ system, defined with respect to the basis $\{\ket{a_i}\}$. Note that $\tilde{\rho}^{AB}$ is symmetric since for Hermitian operators $E$ and $F$,
\[
(E\otimes F)(\tilde{\rho}^{AB}) = \mathrm{Tr}[(E \otimes F)\tilde{\rho}^{AB}] = \sum_{ij} \lambda_i \lambda_j E_{ji} F_{ij} 
= (F\otimes E)(\tilde{\rho}^{AB}).
\]
Note also that $\tilde{\rho}^{AB}$ is positive semi-definite since for any Hermitian operator $E$,
\[
(E\otimes E)(\tilde{\rho}^{AB}) =  \mathrm{Tr}[(E \otimes E)\tilde{\rho}^{AB}] = \sum_{ij} \lambda_i \lambda_j E_{ji} E_{ij} = \sum_{ij} \lambda_i \lambda_j |E_{ji}|^2 \geq  0.
\]
Therefore $\tilde{\rho}^{AB}$ is an inner product state. The quantum state $\rho^{AB}$ can be written $\rho^{AB} = (\id \otimes \tau)(\tilde{\rho}^{AB})$, where $\tau$ is the transpose map followed by $U^{-1}$, which proves the theorem. 
\end{proof}
Now any correlations that can be obtained from measurements on a bipartite classical or quantum system, pure or mixed, can also be obtained from measurements on a pure quantum state of two $d$-dimensional systems for some $d$. This follows from the fact that mixed quantum states always have a purification on a larger Hilbert space. Combining this observation with theorems \ref{equivstatesweaker} and \ref{quantumstatesareinnerprod} gives
\begin{thm}\label{classicalquantumcorrelationsinnerprod}
Any correlations obtained from measurements on a bipartite, pure or mixed, classical or quantum system could also be obtained from measurements on an inner product state.
\end{thm} 
Hence as far as correlations go, the fact that we consider only inner product states is not nearly so restrictive as it looks. By extension, the results apply to all classical and quantum bipartite systems.

\subsection{The set $Q_1$}\label{sec:thesetq1}

The problem of characterizing those correlations which could in principle be produced by performing measurements on quantum systems, and those that cannot, is an interesting one. Tsirelson's inequality, which limits the possible violation of the CHSH inequality in quantum theory, was the first result in this direction. A great deal of progress is made in Refs.~\cite{NPA, momentproblem}, where the problem is reduced to the following form. A hierarchy of sets $Q_1, Q_2,\ldots$ is defined, such that each $Q_k$ is a proper subset of the set of all possible bipartite non-signalling correlations, and each $Q_k$ is strictly contained in its predecessor. For given correlations $P(a,b|x,y)$, and for each $k$, it is a semi-definite programming problem to determine whether $P(a,b|x,y)$ is contained in $Q_k$. Furthermore, a given set of correlations $P(a,b|x,y)$ can be obtained from measurements on quantum systems if and only if $P(a,b|x,y)$ is contained in $Q_k$ for some $k$. Hence the sets $Q_k$ become smaller as $k$ increases, until in the limit $k\rightarrow\infty$ they converge towards the set $Q$ of quantum correlations.

The set $Q_1$, which is the largest in the hierarchy, is of further significance. In Ref.~\cite{ML} it is shown that correlations in $Q_1$ satisfy a readily comprehensible physical principle called \emph{macroscopic locality}. For a precise description of what this means, see Ref.~\cite{ML}, but in a nutshell, the principle states that the coarse-grained statistics of correlation experiments involving a large number of particles should admit a description by a local hidden variable model. In other words, the set of microscopic correlations that satisfy the principle of macroscopic locality are those which are compatible with classical physics in a certain limit in which the number of particle pairs being tested is large, and only coarse-grained statistics, rather than settings and outcomes for every pair, are collected.
It is also known that $Q_1$ is closed under \emph{wiring} \cite{ML,Allcock}, in other words it is not possible to distill correlations in $Q_1$ to correlations outside $Q_1$ by performing measurements on a number of distinct pairs of systems, and locally manipulating the data. Finally, in the specific case of binary measurement choices and outcomes, all correlations in $Q_1$ respect Tsirelson's bound of $2\sqrt{2}$ for the CHSH scenario. The main theorem below states that correlations from measurements on inner product states are contained in the set $Q_1$. 

First, we give a formal definition of $Q_1$. Suppose that Alice and Bob share two systems in a bipartite state, and let Alice choose a measurement $x$ and Bob choose a measurement $y$. Up to now, when we discussed correlations, Alice's and Bob's outcomes were labelled $a$ and $b$, and correlations written $P(a,b|x,y)$. For the specific purpose of defining $Q_1$, however, it is more useful to label the measurement outcomes in such a way that outcomes of distinct measurements have different labels. Hence let the index $i$ range over all possible outcomes of all of Alice's measurement choices. For example, if Alice is choosing from $N$ possible measurements, each of which has $k$ possible outcomes, then $i$ takes values in $\{1,\ldots, kN\}$, with $i=1,\ldots, k$ the outcomes of the $x=1$ measurement, $i=k+1,\ldots, 2k$ the outcomes of the $x=2$ measurement, and so on. Let the same conventions apply to Bob's outcome, which is denoted $j$. With a slight abuse of notation, let $x(i)$ denote the unique measurement choice of Alice for which $i$ is a possible outcome. Similarly, $y(j)$. Write $P(i,j)$ for the probability of obtaining outcomes $i$ and $j$ when the measurements $x(i)$ and $y(j)$ are performed. Let $P_A(i)$ denote the marginal probability for Alice to obtain outcome $i$ when she performs measurement $x(i)$, and $P_B(j)$ denote the marginal probability for Bob to obtain outcome $j$ when he performs measurement $y(j)$.
\begin{dfn}[\cite{NPA,momentproblem,ML}]\label{q1def}
A set of correlations $P(i,j)$ is in $Q_1$ iff there exists a positive semi-definite matrix $\gamma$ of the form
\begin{equation}\label{eq:q1def1}
\gamma =
\begin{pmatrix}
1 & \vec{P}_A^T & \vec{P}_B^T \\
\vec{P}_A & \tilde{Q} & \tilde{P} \\
\vec{P}_B & \tilde{P}^T & \tilde{R}\\
\end{pmatrix},
\end{equation}
such that
\begin{enumerate}
\item $\vec{P}_A$ and $\vec{P}_B$ are the vectors of probabilities $P_A(i)$ and $P_B(j)$,
\item $\tilde{P}$ is a matrix with elements $\tilde{P}_{ij} = P(i,j)$,
\item $\tilde{Q}$ and $\tilde{R}$ are sub-matrices with diagonal elements $\tilde{Q}_{ii} = P_A(i)$ and $\tilde{R}_{jj} = P_B(j)$,
\item $\tilde{Q}_{ii'} = 0$ if $i\ne i'$, $x(i)=x(i')$,
\item $\tilde{R}_{jj'} = 0$ if $j\ne j'$, $y(j)=y(j')$.
\end{enumerate}
In words, the last two conditions state that elements of $\tilde{Q}$ and $\tilde{R}$ corresponding to different outcomes of the \emph{same} measurement must be zero. The remaining off-diagonal elements of $\tilde{Q}$ and $\tilde{R}$ can be chosen freely.   
\end{dfn}

\subsection{The main theorem}
\label{sec:Q1Poly}
\begin{thm}\label{corr:identitystatecorrelationsareinq1}
Consider two similar systems, whose joint state is an inner product state. All correlations that can be obtained from local measurements lie in $Q_1$.
\end{thm}
\begin{proof}
It is sufficient to show that for any set of correlations generated by measurements on an inner product state, there exists a matrix $\gamma$ of the form \eqref{eq:q1def1}, which is symmetric, positive semi-definite, and has the feature that entries in the blocks $\tilde{Q}$ and $\tilde{R}$ corresponding to different outcomes of the same measurement are zero.

Consider correlations generated by measurements on an inner product state $\omega^{AB}$. Using the notation introduced in section \ref{sec:thesetq1}, let $e_i$ be the effect corresponding to Alice's measurement outcome $i$, and $f_j$ the effect corresponding to Bob's measurement outcome $j$. Suppose that $i$ ranges from $1,\ldots,n^A$ and $j$ from $1,\ldots,n^B$. Define a vector of effects $g = (u,e_1,\dots,e_{n^A},f_1,\dots,f_{n^B})$, and denote the entries $g_1=u, g_2=e_1,\ldots,g_{1+n^A+n^B}=f_{n^B}$. Define the $(1+n^A+n^B)\times(1+n^A+n^B)$ matrix $\tilde{\gamma}$ such that $\tilde{\gamma}_{kl} = (g_k \otimes g_l)(\omega^{AB})$. From the fact that $\omega^{AB}$ is an inner product state, it follows directly that $\tilde{\gamma}$ is a symmetric and positive semi-definite matrix \cite{Bhatia}.

Now define a matrix $\gamma$ of the form \eqref{eq:q1def1}, with $\gamma_{kl} = \tilde{\gamma}_{kl}$ for all $k,l$ except for the following elements of the sub-matrices $\tilde{Q}$ and $\tilde{R}$:
\begin{enumerate}
\item $\tilde{Q}_{ii} = P_A(i)$, and $\tilde{R}_{jj} = P_B(j)$.
\item $\tilde{Q}_{ii'} = 0$ if $i\ne i'$, $x(i)=x(i')$,
\item $\tilde{R}_{jj'} = 0$ if $j\ne j'$, $y(j)=y(j')$.
\end{enumerate}
By construction, $\gamma$ satisfies conditions (i)-(v) of Definition \ref{q1def}, and symmetry of $\gamma$ follows from symmetry of $\tilde{\gamma}$. It remains to show that $\gamma$ is positive semi-definite.

To this end, let $\delta = \gamma - \tilde{\gamma}$ and note that $\delta$ is of the form
\begin{equation}
\delta =
\begin{pmatrix} 0 & \cdots & 0 \\
\vdots & \delta_Q & \tilde{0} \\
0 & \tilde{0}^T & \delta_R \\
\end{pmatrix},
\end{equation}
where $\delta_Q$ is an $n_A\times n_A$ sub-matrix, $\delta_R$ is an $n_B\times n_B$ sub-matrix, and $\tilde{0}$ is the $n_A \times n_B$ matrix with all entries $0$. Since both $\gamma$ and $\tilde{\gamma}$ are symmetric, $\delta$ is also symmetric. We will show that $\delta_Q$ and $\delta_R$ are positive semi-definite. It follows that $\delta$ is positive semi-definite. Since $\gamma = \delta + \tilde{\gamma}$, it follows that $\gamma$ is also positive semi-definite.

Note that $(\delta_Q)_{ii'} = 0$ for $x(i)\ne x(i')$. It follows that $\delta_Q$ is block diagonal, with each block corresponding to a particular measurement choice of Alice. Consider a particular block, corresponding to a measurement with, say, $r$ outcomes. It is of the form
\begin{equation}
M = \begin{pmatrix}
e_1\otimes u - e_1 \otimes e_1 & -e_1 \otimes e_2 & \cdots & -e_1\otimes e_r \\
- e_2 \otimes e_1 & e_2 \otimes u - e_2\otimes e_2 & \cdots & -e_2\otimes e_r \\
 & & \vdots & \\
 -e_r \otimes e_1 & -e_r \otimes e_2 & \cdots & e_r \otimes u - e_r \otimes e_r
\end{pmatrix}
(\omega^{AB}).
\end{equation}
Using $e_1 + \cdots + e_r = u$, this matrix can be decomposed into a sum of $(r^2-r)/2$ matrices
\begin{equation}
  M = \sum_{n=2}^r \sum_{m=1}^{n-1} M^{mn},
\end{equation}
where all entries of the matrices $M^{mn}$ are $0$, except for
\begin{align}
(M^{mn})_{mm} &= (M^{mn})_{nn} = (e_m\otimes e_n)(\omega^{AB}) \\
(M^{mn})_{mn} &= (M^{mn})_{nm} = -(e_m\otimes e_n)(\omega^{AB}).
\end{align}
Each $M^{mn}$ is manifestly positive semi-definite, hence $M$ is positive semi-definite. Since each block of $\delta_Q$ is positive semi-definite, $\delta_Q$ is also positive semi-definite. A similar argument shows that $\delta_R$ is also positive semi-definite. Therefore $\delta$ and $\gamma$ are positive semi-definite. This concludes the proof. 
\end{proof}

\begin{cor}\label{corollary}
Consider two systems, whose joint state is of the form $\omega^{AB} = (\id \otimes \tau) (\sigma^{AB})$, where $\tau: V_+\rightarrow V_+$ takes normalized states to normalized states and $\sigma^{AB}$ is an inner product state. All correlations obtainable from measurements on $\omega^{AB}$ lie in $Q_1$.
\end{cor}  
\begin{proof}
This is immediate from theorem \ref{corr:identitystatecorrelationsareinq1} and theorem \ref{equivstatesweaker}.
\end{proof}
Theorem \ref{classicalquantumcorrelationsinnerprod} then implies that all correlations from bipartite classical and quantum states lie in $Q_1$. This was known already of course from Refs.~\cite{NPA, momentproblem}. One could view the theorem and corollary as an independent proof of this fact.

\section{Polygons revisited}\label{polygonsrevisited}

It has already been observed that given two $n$-vertex polygon systems, the maximally entangled state $\phi^{AB}$, defined in section \ref{sec:defpolybox}, is an inner product state if and only if $n$ is odd. Theorem \ref{corr:identitystatecorrelationsareinq1} states that correlations obtained from measurements on an inner product state lie in the set $Q_1$, which means in particular that they respect Tsirelson's bound for the CHSH inequailty. This explains why Tsirelson's bound is satisfied by the odd $n$ polygon systems, and is consistent with violation of Tsirelson's bound by the even $n$ polygon systems.

This section relates these observations to simple geometrical properties of the state spaces of polygon systems. A quick glance at figures \ref{model} and \ref{fig:poly3d} reveals an obvious difference between the odd $n$ and even $n$ cases. For odd $n$, the effect cone $V_+^*$ coincides with the state cone $V_+$. For even $n$ on the other hand, the effect cone is isomorphic to the state cone, but is rotated by some non-zero angle. This simple observation lies at the heart of why it is only the maximally entangled states of odd $n$ polygon systems that are inner product states, and hence why it is only these that must satisfy Tsirelson's bound.   

The fundamental difference between the odd $n$ and even $n$ state spaces can be stated more formally as follows. First
\begin{dfn}[weakly self-dual]\label{def:weaklyselfdual}
A system is \emph{weakly self-dual} iff the state and effect cones are isomorphic.
\end{dfn}
All of the polygon state spaces are weakly self-dual. The isomorphisms are simply the rotations and improper rotations around the $z$ axis by $(1 + 2k) \pi/n,\ k\in\{0,\dots,n-1\}$ if $n$ is even and by $2k \pi/n,\ k\in\{0,\dots,n-1\}$ if $n$ is odd. 

The odd $n$ polygon state spaces, on the other hand, satisfy a stronger condition, whereby there are additional restrictions on the isomorphism connecting $V_+^*$ and $V_+$.
\begin{dfn}[strongly self-dual]\label{def:stronglyselfdual}
A system is \emph{strongly self-dual} iff there exists an isomorphism $T: V^*_+ \to V_+$ which is symmetric and positive semi-definite, i.e., $f[T(e)] = e[T(f)]$ for all $e,f \in V^*$, and $e[T(e)] \geq 0$ for all $e \in V^*$. 
\end{dfn}
Given the representation of sections \ref{polygonsystems} and \ref{sec:defpolybox}, the identity map is an example of such an isomorphism. The odd $n$ polygon state spaces are strongly self-dual, but the even $n$ are not.

The concepts of strong and weak self-duality have appeared earlier in the literature, for example in Ref.~\cite{teleport}. Weak self-duality is intimately related to the operational tasks of probabilistic remote state preparation (steering) and teleportation \cite{steering,teleport}.

Now we can relate these properties of individual systems to the bipartite maximally entangled state $\phi^{AB}$. Notice that given two similar systems, any isomorphism $T: V^*_+ \to V_+$ corresponds to a bipartite state $\omega_T^{AB}$ via
\begin{equation}
(e \otimes f)(\omega_T^{AB}) = \frac{f[T(e)]}{u[T(u)]} .
\end{equation}
The state defined is normalized by construction and is locally positive since $0 \leq f[T(e)]/u[T(u)] \leq 1$ for all $e,f\in E(\Omega)$. Intuitively, $\omega_T^{AB}$ is defined so that if Alice performs a measurement and obtains outcome $e$, then Bob's unnormalized collapsed state, conditioned on that outcome, is $T(e)$. 

In the special case that the individual systems are strongly self-dual and the isomorphism $T$ has the additional properties required by definition \ref{def:stronglyselfdual}, then the induced state $\omega_T^{AB}$ is symmetric and positive semi-definite, hence it is an inner product state. This is the case for the maximally entangled state $\phi^{AB}$ of odd $n$ polygon systems, defined in \eref{definitionphi}, where $\phi^{AB}$ corresponds to a map $T$ which is simply the identity map. It follows that for odd $n$, correlations from $\phi^{AB}$ lie in $Q_1$. 

In the case that individual systems are weakly but not strongly self-dual, the maximally entangled state corresponds to an isomorphism $T$, but there is no such $T$ with the additional properties of symmetry and positive semi-definiteness, hence the maximally entangled state is not an inner product state. This is the case for the maximally entangled state $\phi^{AB}$ of the even $n$ polygon systems, defined in \eref{definitionphi}, where $\phi^{AB}$ corresponds to a map $T$ which is a rotation in $\mathbb{R}^3$ by $\pi/n$. This is why for even $n$, correlations from $\phi^{AB}$ need not lie in $Q_1$.

\section{Correlations outside of $Q_1$}\label{sec:Correlationsoutsideofq1}

Correlations obtained from the maximally entangled state of two odd $n$ polygon systems must be contained in $Q_1$, and this has been seen to be related to the fact that the individual systems are strongly self-dual. It is natural to ask whether the correlations obtained from \emph{any} joint state of strongly self-dual subsystems must also lie in $Q_1$. An explicit counterexample shows that this is not the case.

Consider a strongly self-dual system with normalized extremal states
\begin{align*}
 \omega_1&=(1, 0, 1)^T & \omega_2&=(0, 1, 1)^T & \omega_3&=(-1, 0, 1)^T\\
 \omega_4&=(-1, -1, 1)^T & \omega_5&=(1, -1, 1)^T ,
\end{align*}
and normalized ray extremal effects
\begin{align*}
 e_1 &= \frac{1}{2} (1, 0, 1)^T & e_2 &= \frac{1}{2} (0, 1, 1)^T & e_3 &= \frac{1}{2} (-1, 0, 1)^T\\
 e_4 &= \frac{1}{3} (-1, -1, 1)^T & e_5 &= \frac{1}{3} (1, -1, 1)^T & u &= (0,0,1)^T .
\end{align*}
The state space for this system looks something like a house and is depicted in \fref{fig:house}. 
\begin{figure}[tb]
 \centering
 \includegraphics[width=0.35 \linewidth]{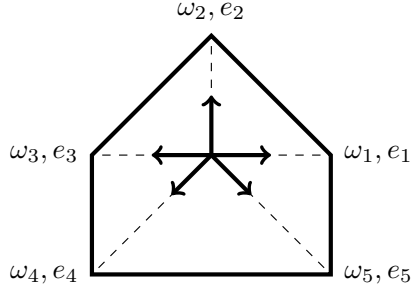}
 \caption{The house-shaped state space is strongly self-dual.}
 \label{fig:house}
\end{figure}

We have explicitly calculated all extremal states in the maximal tensor product of two such systems.
One of these joint states can be written as
\begin{equation}
  \left(\begin{array}{ccc} -1 & -\frac{1}{4} & -\frac{1}{2} \\ \frac{1}{4} & -\frac{1}{2} & -\frac{1}{4} \\ \frac{1}{2} & -\frac{1}{4} & 1\end{array}\right) ,
\end{equation}
where we have used the same representation as a $3\times 3$ matrix that was introduced in section~\ref{sec:defpolybox}. 
This state is extremal in the maximal tensor product, but is not an inner product state.
With a suitable choice of measurements, correlations can be produced which violate Uffink's quadratic inequality \cite{Uff02}
\begin{equation}\label{eq:uffink}
(E_{0,0}+E_{1,0})^2 + (E_{0,1}-E_{1,1})^2 \leq 4.
\end{equation}
In particular the measurement choices
\begin{align}
  x = 0:& \{e_5, u-e_5\}, & x=1:& \{e_3, u-e_3\}, & y=0:& \{e_2, u-e_2\}, & y=1:& \{e_3, u-e_3\}
\end{align}
give
\begin{equation}
(E_{0,0}+E_{1,0})^2 + (E_{0,1}-E_{1,1})^2 = \frac{17}{4} > 4 .
\end{equation}
However, satisfaction of Uffink's inequality is known to be a necessary condition for membership of $Q_1$
\cite{QMbound}; hence these correlations cannot lie in $Q_1$.

Although these correlations violate Uffink's inequality and lie outside of $Q_1$, they do not violate Tsirelson's bound for the CHSH inequality. In fact, we have not been able to find a joint state of two strongly self-dual subsystems that violates the CHSH inequality beyond Tsirelson's bound. This leads us to conjecture that Tsirelson's bound holds for every theory with strongly self-dual subsystems.

\section{Discussion}\label{discussion}
One way of viewing the difference between classical and quantum systems is that the structure, or shape, of the space of possible states of a system is different. For example in the case of a classical trit, the state space is the space of probability distributions over trit values, which is geometrically a triangle. In the case of a qubit, the state space is the Bloch ball. This work considers a very general setting in which a whole range of probabilistic models can be defined, with the classical and quantum theories as special cases. There is little constraint on the state space, except that it is assumed to be convex, and joint systems are assumed to satisfy a no-signalling principle and a principle of local tomography. The aim is to investigate the nonlocal correlations that can be produced by measurements on entangled systems in these models, and to compare and contrast with the classical and quantum cases. 

The main theorem, with its corollary, states that correlations from a broad class of bipartite states in probabilistic theories cannot be arbitrarily nonlocal --- they are constrained to obey the principle of \emph{macroscopic locality}, or equivalently to lie within the set $Q_1$, which means in particular that they satisfy Tsirelson's bound for violation of the CHSH inequality. This theorem extends to all bipartite quantum states, which explains why quantum mechanics cannot violate macroscopic locality or Tsirelson's bound. 

The work has also revealed an intimate and intricate relationship between the shape of the state space for an individual system, and the strength of the nonlocal correlations that can be obtained from two systems in an entangled state. This is illustrated by a family of models, in each of which the state space for a single system is a regular polygon with $n$ vertices. Given two such systems, there is an analogue of a maximally entangled state. It turns out that the strength of nonlocal correlations generated by this state depends dramatically on the parity of the number of vertices $n$ of the local polygon. If $n$ is even, maximally nonlocal correlations can be generated, including those that violate macroscopic locality. If $n$ is odd, however, the maximally entangled state respects macroscopic locality. This is in turn explained by the fact that odd $n$ polygons have a geometric property known as strong self-duality, while even $n$ polygons do not.

It would be natural to think that \emph{all} bipartite states of strongly self-dual subsystems would respect macroscopic locality, but the house-shaped counterexample shows that this is not the case. An interesting open question, therefore, is the following: What additional property of local state spaces would ensure that all bipartite states give correlations which respect macroscopic locality? One suggestion is the constraint that for any ray extremal effect, there is a unique state on which this effect will occur with certainty. This property is very attractive from a physical point of view. It allows a natural definition of the post-measurement states of these effects, such that repeating a measurement reproduces the same outcome. This extra constraint is indeed not satisfied by the house model, since the effect $e_1$ occurs with certainty for both states $\omega_1$ and $\omega_5$, but it is satisfied by odd $n$ polygon models.
Another possibility that seems to be plausible is that strong self-duality together with the property that all extremal states of the local systems can be transformed into one another reversibly might limit the set of possible correlations to the ones compatible with macroscopic locality.

Finally, it is worth emphasizing that two theories which have almost identical local state spaces can lead to dramatically different nonlocal correlations. In particular, given any finite level of accuracy, it is always possible to find a polygon model with an even and sufficiently large number of vertices $n$, which is locally indistinguishable from the quantum-like case, where the state space is a disc. Nevertheless, while quantum correlations are restricted, any non-signalling correlations can be distilled in the former model by using multiple copies of the maximally entangled state.

\ack
We thank Andreas Winter, Volkher Scholz, Markus M\"uller and Cyril Branciard for insightful discussions. JB is supported by an EPSRC Career Acceleration Fellowship. We acknowledge financial support from the German National Academic Foundation. NB is supported by the UK EPSRC.

\appendix
\setcounter{section}{0}

\section{Optimal CHSH value}
\label{app:optimalchshsvalue}
In the main text, we gave expressions for the maximal CHSH value returned by measurements on a maximally entangled state of two $n$-vertex polygon systems. The expression for even $n$ is given in \eqref{CHSHeven}, and for odd $n$, in \eref{CHSHodd}. The choice of angles that maximize these quantities is not unique.
We will see below that we have to take two different sets of optimal angles into account.
\begin{table}[h]
\caption{\label{tab:idealanglesodd}Optimal angles}
\begin{indented} \item[]
  \begin{tabular}{ccccc}
   \br
   & $\alpha^*_0$ & $\alpha^*_1$ & $\beta^*_0$ & $\beta^*_1$\\
   \mr
   Set $1$ & $0$ & $\frac{\pi}{2}$ & $\frac{\pi}{4}$ & $-\frac{\pi}{4}$\\\ms
   Set $2$ & $0$ & $\frac{\pi}{2}$ & $-\frac{3\,\pi}{4}$ & $\frac{3\,\pi}{4}$\\
   \br
  \end{tabular}
 \end{indented}
 \end{table}

 Note that the optimization has been performed without any restriction on the values of the angles $\alpha^*_x$ and $\beta^*_y$. However, due to the polygon structure of our model, only specific angles, corresponding to extremal effects, are admissible. Thus the optimal CHSH values are obtained by taking the extremal effects which are closest to the optimal angles.
\begin{table}[t]
\caption{\label{tab:chshpoly} Analytical expression for the maximal CHSH-violation of polygon boxes}
\footnotesize
\begin{tabular}{ccccc}
 \br
 $x$		&	$\Delta\alpha_1$	&	$\Delta\beta_0$		&	$\Delta\beta_1$	&	$S$\\
 \mr
 $0$ \rule{0pt}{4mm}	&	$0$			&	$\frac{\pi}{n}$			&	$\frac{\pi}{n}$			&	$2 \sqrt{2}$\\[2mm]
 $1$		&	$\frac{-\pi}{2 n}$	&	$\frac{-\pi}{4 n}$	&	$\frac{\pi}{4 n}$	&	 $\frac{2}{\left(1+\sec\left(\frac{\pi}{n}\right)\right)^2} \, \left[ 1 + \sec\left(\frac{\pi}{n}\right) \left( 2 \cos\left(\frac{n+3}{4 n} \, \pi\right)+6 \sin\left(\frac{n+1}{4 n} \, \pi\right) + \sec\left(\frac{\pi}{n}\right) - 2 \right)\right]$\\[2mm]
 $2$		&	$\frac{\pi}{n}$		&	$\frac{\pi}{2 n}$	&	$\frac{-\pi}{2 n}$	&	$\sec\left(\frac{\pi}{n}\right) \, \left[ 3 \cos(\frac{n+2}{4 n} \,  \pi) + \sin\left(\frac{n+6}{4n} \, \pi\right) \right]$\\[2mm]
 $3$		&	$\frac{\pi}{2 n}$	&	$\frac{\pi}{4 n}$	&	$\frac{-\pi}{4 n}$	&	$\frac{-2}{\left(1+\sec(\frac{\pi}{n})\right)^2} \left[1 - \sec\left(\frac{\pi}{n}\right) \left(6 \cos\left(\frac{n+1}{4 n} \, \pi\right)+ 2 \sin\left(\frac{n+3}{4 n} \, \pi\right)- \sec\left(\frac{\pi}{n}\right)\right)\right]$\\[2mm]
 $4$		&	$0$			&	$0$	&	$0$	&	$2 \sqrt{2} \, \sec(\frac{\pi}{n})$\\[2mm]
 $5$		&	$\frac{-\pi}{2 n}$	&	$\frac{-\pi}{4 n}$	&	$\frac{\pi}{4 n}$	&	$\frac{-2}{\left(1+\sec(\frac{\pi}{n})\right)^2} \left[1 - \sec\left(\frac{\pi}{n}\right) \left(6 \sin\left(\frac{n+1}{4 n} \, \pi\right)+ 2 \cos\left(\frac{n+3}{4 n} \, \pi\right)- \sec\left(\frac{\pi}{n}\right)\right)\right]$\\[2mm]
 $6$		&	$\frac{\pi}{n}$		&	$\frac{-\pi}{2 n}$	&	$\frac{\pi}{2 n}$	&	$\sec(\frac{\pi}{n}) \, \left[ \cos\left(\frac{n+6}{4 n} \, \pi\right) + 3 \sin\left(\frac{n+2}{4n} \, \pi\right) \right]$\\[2mm]
 $7$		&	$\frac{\pi}{2 n}$	&	$\frac{\pi}{4 n}$	&	$\frac{-\pi}{4 n}$	&	 $\frac{2}{\left(1+\sec\left(\frac{\pi}{n}\right)\right)^2} \, \left[ 1 + \sec\left(\frac{\pi}{n}\right) \left( 2 \sin\left(\frac{n+3}{4 n} \, \pi\right)+6 \cos\left(\frac{n+1}{4 n} \, \pi\right) + \sec\left(\frac{\pi}{n}\right) - 2 \right)\right]$\\[2mm]
 \br
\end{tabular}
\end{table}
The deviation from the optimal angles will be called $\Delta\alpha_0, \Delta\alpha_1, \Delta\beta_0, \Delta\beta_1$. Without loss of generality we set $\Delta\alpha_0$ to $0$. A detailed analysis reveals a total of eight classes of deviation angles characterized by the remainder $x = n \mod 8$ of the division of $n$ by $8$. For a free choice of angles both sets in \tref{tab:idealanglesodd} lead to the same maximum value of the CHSH-coefficient. Whether the available extremal effects are closer to the angles of set 1 or set 2, however, depends on the number of vertices. It turns out that for even $n$ as well as for $x \in \{1,7\}$ this is the case for set $1$, whereas for $x \in \{3,5\}$ the smallest derivation can be achieved to set $2$. The maximal CHSH value for each polygon system is given by the following parameters for \eref{CHSHeven} and \eref{CHSHodd}:
\begin{eqnarray*}
 \beta_y &= \beta^*_y + \Delta\beta_y\\
 \alpha_x &= \alpha^*_x + \Delta\alpha_x
\end{eqnarray*}

The eight classes can clearly be seen in \fref{fig:chshpoly}. The analytic expressions for the maximal CHSH value as a function of the number of vertices $n$ and the remainder $x$ are given in \tref{tab:chshpoly}.

\end{document}